\documentclass[a4paper,aps,pra,reprint,floatfix,showpacs]{revtex4-1}

\usepackage{amsfonts}
\usepackage{amssymb}
\usepackage{amsmath}
\usepackage{amstext}
\usepackage{amsthm}
\usepackage{bm}
\usepackage{microtype}
\usepackage{tikz}

\newcommand{\bmr}[1]{\bm{#1}}
\newcommand{\I}{\text{i}}
\newcommand{\E}{\text{e}}
\newcommand{\cc}[1]{{#1}^{*}}
\newcommand{\ket}[1]{|#1 \rangle}

\newcommand{\sym}[1]{| \text{S}_{#1} \rangle}
\newcommand{\psis}{| \psi^{\text{s}} \rangle}
\newcommand{\phis}{| \phi^{\text{s}} \rangle}
\newcommand\dash{\nobreakdash-\hspace{0pt}}
\newcommand{\matha}{\mathcal{A}}
\newcommand{\mathb}{\mathcal{B}}
\newcommand{\mathd}{\mathcal{D}}
\newcommand{\mathm}{\mathcal{M}}
\newcommand{\mbbc}{\mathbb{C}}
\newcommand{\cext}{\hat{\mathbb{C}}}

\newcommand{\mbbrr}{\mathbb{R}^{3}}
\newcommand{\suc}{\text{SU}(2)}
\newcommand{\slc}{\text{SL}(2,\mathbb{C})}
\newcommand{\pslc}{\text{PSL}(2,\mathbb{C})}

\providecommand{\Eq}[1]{Eq.\ \eqref{#1}}
\providecommand{\Fig}[1]{Fig.\ \ref{#1}}
\providecommand{\Sect}[1]{Section~\ref{#1}}
\providecommand{\Theo}[1]{Theorem~\ref{#1}}

\newtheorem{theorem}{Theorem}

\newtheorem{corollary}[theorem]{Corrolary}

\begin{document}

\title{Symmetric entanglement classes for $n$ qubits}

\author{Martin Aulbach}
\email{m.aulbach1@physics.ox.ac.uk}
\affiliation{Department of Physics, University of Oxford, Clarendon
  Laboratory, Oxford OX1 3PU, United Kingdom}
\affiliation{The School of Physics and Astronomy, University of Leeds,
  Leeds LS2 9JT, United Kingdom}

\begin{abstract}
  Permutation-symmetric $n$ qubit pure states can be represented by
  $n$ points on the surface of the unit sphere by means of the
  Majorana representation. Here this representation is employed to
  characterize and compare the three entanglement classification
  schemes LOCC, SLOCC and the Degeneracy Configuration.  Symmetric
  SLOCC operations are found to be described by M\"{o}bius
  transformations, and an intuitive visualization of their freedoms is
  presented.  For symmetric states of up to 5 qubits explicit forms of
  representative states for all SLOCC classes are derived. The
  symmetric 4 qubit entanglement classes are compared to the
  entanglement families introduced in [PRA \textbf{65}, 052112
  (2002)], and examples are given how the SLOCC-inequivalence of
  symmetric states can be quickly determined from known results about
  M\"{o}bius transformations.
\end{abstract}
\pacs{03.67.Mn, 03.65.Ud, 02.40.Tt, 02.40.Dr}
\maketitle

\section{Introduction}\label{introduction}

Multipartite entanglement is an essential resource in quantum
information science, and therefore it is desirable to categorize the
states of a given Hilbert space into groups of states with similar
entanglement.  The object of interest in this paper are
permutation-symmetric states. These kind of states have prominently
featured in several recent works, such as the characterization of
SLOCC entanglement classes \cite{Markham10,Bastin09,Mathonet10}, the
determination of maximal entanglement in terms of the geometric
measure \cite{Aulbach10,JMartin,Aulbach11}, uses for entanglement
witnesses or in experimental setups
\cite{Korbicz05,Korbicz06,Prevedel09,Wieczorek09}, for finding
solutions to the Lipkin-Meshkov-Glick model \cite{Ribeiro08}, and for
quantifying the ground state entanglement of the same model
\cite{Orus08}.

The central tool in all of these studies was the Majorana
representation \cite{Majorana32}, a generalization of the Bloch sphere
representation which allows symmetric $n$ qubit states to be uniquely
represented by $n$ undistinguishable points on the sphere.  Here this
paradigm is employed to discuss three different entanglement
classification schemes, namely, LOCC, SLOCC and the recently
introduced Degeneracy Configuration \cite{Bastin09}, for symmetric $n$
qubit states.  It is seen that symmetric SLOCC operations can be
described by the M\"{o}bius transformations of complex analysis, a
result that is not only of theoretical interest but also of practical
value, e.g., to determine whether two symmetric states belong to the
same SLOCC class.  Intriguingly, SLOCC operations can be uniquely
decomposed into affine M\"{o}bius transformations and LOCC operations,
thus allowing for a straightforward visualization of the innate SLOCC
freedoms.  A study of all symmetric SLOCC and DC classes for up to 5
qubits will yield the analytical form of representative states for
each class. For the 4 qubit case the results are put into relation to
the concept of entanglement families introduced in
\cite{Verstraete02}.

\section{Majorana Representation}\label{majrep}

Permutation-symmetric quantum states are defined as being invariant
under any permutation of their subsystems.  For an $n$-partite state
$\ket{\psi}$ this is the case iff $P \ket{\psi} = \ket{\psi}$ for all
$P \in S_{n}$, where $S_{n}$ is the symmetric group of $n$ elements.
For $n$ qubits the symmetric sector of the Hilbert space is spanned by
the $n+1$ Dicke states $\sym{k}$, $0 \leq k \leq n$, the equally
weighted sums of all permutations of computational basis states with
$n-k$ qubits being $\ket{0}$ and $k$ being $\ket{1}$:
\begin{equation}\label{dicke_def}
  \sym{k} = {\binom{n}{k}}^{- 1/2} \sum_{\text{perm}}
  \underbrace{ \ket{0} \ket{0} \cdots \ket{0} }_{n-k}
  \underbrace{ \ket{1} \ket{1} \cdots \ket{1} }_{k} \, .
\end{equation}
By means of the Majorana representation any permutation\dash symmetric
state $\psis$ of $n$ spin-$\frac{1}{2}$ particles can be uniquely
represented, up to an unphysical global phase, by a multiset of $n$
points on $S^{2}$, with an isomorphism mediating between the pure
states of the symmetric subspace and the set of $n$ unit vectors in
$\mbbrr$ \cite{Majorana32,Bacry74}.  Mathematically, this is expressed
as
\begin{equation}\label{majorana_definition}
  \psis = \frac{\E^{\I \delta}}{\sqrt{K}}
  \sum_{ \text{perm} } \ket{\phi_{P(1)}} \ket{\phi_{P(2)}}
   \cdots \ket{\phi_{P(n)}} \, ,
\end{equation}
where $\E^{\I \delta}$ is a global phase, $K$ the normalization
factor, and the sum runs over all permutations of $n$ single qubit
states $\ket{\phi_i} = \cos \frac{\theta_i}{2} \ket{0} + \E^{\I
  \varphi_i} \sin \frac{\theta_i}{2} \ket{1}$.  Thus the multi-qubit
state $\psis$ can be visualized by $n$ Bloch vectors $\ket{\phi_i}$ on
the surface of a sphere.  These points are called the Majorana points
(MP), and the sphere is called the Majorana sphere.  See, for example,
\cite{Markham10,Aulbach10} for some examples of Majorana
representations.

By means of a stereographic projection the MPs can be projected from
the sphere onto the complex plane, where they coincide with the roots
of the Majorana polynomial
\begin{equation}\label{majpoly}
  \psi (z) = \sum_{k=0}^{n} (-1)^{k-n} a_{k} \sqrt{\tbinom{n}{k}} \,
  z^{k} \: \propto \: \prod_{i=1}^{n} ( z - z_{i} ) \enspace .
\end{equation}
The function $\psi (z)$ represents symmetric states in terms of spin
coherent states \cite{Kolenderski08}, and is also known as the
characteristic polynomial, amplitude function \cite{Radcliffe71}, or
coherent state decomposition \cite{Leboeuf91}.

\section{Entanglement Classes}\label{entclasses}

In order to categorize different types of entanglement, the given
Hilbert space can be partitioned into equivalence classes.  For LOCC
operations the equivalence classes contain those states that can be
deterministically interconverted by means of local operations and
classical communication.  A coarser partition is achieved by SLOCC
equivalence, which is mediated by stochastic LOCC operations
\cite{Bennett00,Mathonet10}. In the symmetric sector, a yet more
coarse partition is the \textit{Degeneracy Configuration} (DC), which
depends on the number of coinciding MPs of symmetric states
\cite{Bastin09}.  These three entanglement classification schemes will
now be outlined.

\subsection{LOCC}

It is known (Corollary 1 of \cite{Bennett00}) that two states are
LOCC-equivalent iff they are LU-equivalent. For multiqubit symmetric
states the condition for LOCC-equivalence of two states $\psis$ and
$\phis$ reads:
\begin{equation}\label{locccond}
  \psis \stackrel{\text{LOCC}}{\longleftrightarrow} \phis
  \Leftrightarrow \exists \, \matha \in \suc :
  \psis = \matha^{\otimes n} \phis .
\end{equation}
The LU can be restricted to the form $\matha^{\otimes n}$, because
there always exists a fully symmetric LU that mediates between two
LOCC-equivalent symmetric states \cite{Mathonet10,Cenci11}.  The
special unitary group $\suc$ has 3 real degrees of freedom (d.f.)
that can be identified with the three rotation axes on the Bloch
sphere. The effect of $\matha^{\otimes n}$ on $\psis$ can then be
understood as a rotation of the Majorana sphere which changes the
location of MPs, but leaves the relative MP distribution
(i.e., distances and angles) intact \cite{Markham10}.

\subsection{SLOCC}

SLOCC operations are mathematically expressed as invertible local
operations \cite{Dur00}.  In the case of two $n$ qubit symmetric
states $\psis$ and $\phis$, the condition for SLOCC-equivalence can be
cast as:
\begin{equation}\label{slocccond}
  \psis \stackrel{\text{SLOCC}}{\longleftrightarrow} \phis
  \Leftrightarrow \exists \, \mathb \in \slc :
  \psis = \mathb^{\otimes n} \phis .
\end{equation}
This operation can be chosen to be fully symmetric \cite{Mathonet10},
and from \Eq{majorana_definition} it is clear that $\mathb$ acts on
each MP individually.  In the following we will therefore always
consider single-qubit operations $\mathb$ (or $\matha$) instead of the
tensor product $\mathb^{\otimes n}$ (or $\matha^{\otimes n}$) whenever
referring to symmetric SLOCC (or LOCC) operations.  The special linear
group $\slc$ which contains the $2 \times 2$ complex matrices with
unit determinant has six real d.f., and because of $\suc \subset \slc$
three of them can be identified as rotations of the Bloch sphere.  The
Lie group $\slc$ is a double cover of the M\"{o}bius group, the
automorphism group on the Riemann sphere. Therefore the
transformations of MPs under symmetric SLOCC operations are described
by the M\"{o}bius transformations of complex analysis, with the
Majorana sphere in lieu of the Riemann sphere (see \Fig{ghztrafos} and
\Fig{sterproj}).  The concept of M\"{o}bius transformations will be
outlined in detail in \Sect{mobius}.

\begin{figure}
  \includegraphics{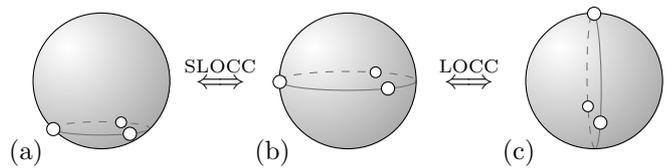}
  \caption{\label{ghztrafos} MPs (shown as white dots) of three
    different symmetric states of 3 qubits. The GHZ state $\sym{0} +
    \sym{3}$, displayed in (b), is LOCC-equivalent to the state
    $\sym{0} + \sqrt{3} \sym{2}$ shown in (c).  The GHZ-like state
    $\alpha \sym{0} + \beta \sym{3}$ in (a) is not LOCC\protect\dash
    equivalent to the others, but can be reached by a SLOCC
    operation.}
\end{figure}

\subsection{Degeneracy Configuration}

The Degeneracy Configuration (DC) of a symmetric $n$ qubit state is
characterized by the number of coinciding MPs \cite{Bastin09}.  The DC
class $\mathd_{n_1 , \ldots , n_d}$ with $n = n_1 + \ldots + n_d$
($n_1 \geq \ldots \geq n_d$) encompasses those states where $n_1$ MPs
coincide on one point of the Bloch sphere, $n_2$ on a different point,
and so on. The number $d$ is called the \textit{diversity degree}, and
the number of DC classes for $n$ qubit symmetric states is given by
the partition function $p(n)$.  The DC class of a given symmetric
state does not change under symmetric SLOCC operations, because of the
automorphism nature of the M\"{o}bius group. On the other hand, two
states that belong to the same DC class do not necessarily belong to
the same SLOCC class \cite{Bastin09}.

\subsection{Hierarchy of classification schemes}

Given two partitions $A$ and $B$ of a set $M$, the partition $A$ is
called a \textit{refinement} of $B$ ($A \leq B$) if every element of
$A$ is a subset of some element of $B$.  Since LOCC is a special case
of SLOCC, and because the DC is invariant under SLOCC, the following
statement can be made:
\begin{theorem}\label{hierarchy}
  The symmetric subspace of every pure \emph{n} qubit Hilbert space
  has the following refinement hierarchy of entanglement partitions:
  \begin{equation}
    \text{LOCC} \leq \text{SLOCC} \leq \text{DC} \enspace .
  \end{equation}
\end{theorem}

\section{M\"{o}bius Transformations}\label{mobius}

As outlined in the previous section, SLOCC operations between
multiqubit symmetric states can be understood as M\"{o}bius
transformations.  These isomorphic functions $f : \cext \rightarrow
\cext$ are defined on the extended complex plane $\cext = \mbbc \cup
\{ \infty \}$ as the rational functions
\begin{equation}\label{mobiusform}
  f(z) = \frac{az + b}{cz + d} \enspace ,
\end{equation}
with $a, b, c, d \in \mbbc$, and $ad - bc \neq 0$. The latter
condition ensures that $f$ is invertible.  The coefficients give rise
to the matrix representation $\mathb = \left(
  \begin{smallmatrix}
    a & b \\
    c & d
  \end{smallmatrix}
\right)$ of the M\"{o}bius group, and from \Eq{mobiusform} it is clear
that it suffices to consider those $\mathb$ with determinant one
(i.e., $ad - bc = 1$).  Since $+\mathb$ and $-\mathb$ desribe the same
transformation $f$, the M\"{o}bius group is isomorphic to the
projective special linear group $\pslc = \slc / \{ \pm \openone \}$.

\begin{figure}
  \includegraphics{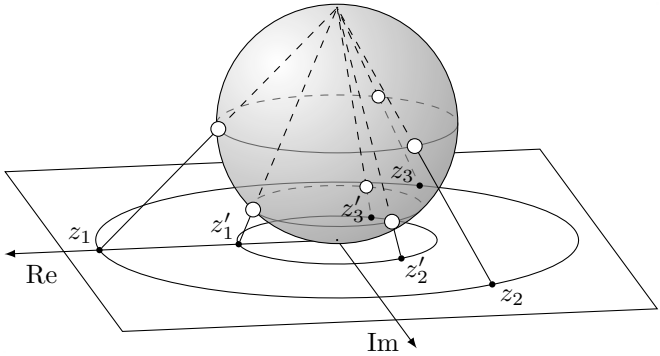}
  \caption{\label{sterproj} A stereographic projection through the
    north pole of the Majorana sphere mediates between the Majorana
    roots in the complex plane and the MPs on the surface of the
    sphere. The SLOCC operation of \Fig{ghztrafos} is facilitated by
    the transformation $f(z) = z/2$ which maps the set of roots $\{
    z_1 , z_2 , z_3 \}$ onto the set $\{ z'_1 , z'_2 , z'_3 \}$, thus
    lowering the ring of MPs. On the sphere M\"{o}bius transformations
    always project circles onto other circles \protect\cite{Knopp}.}
\end{figure}

By means of an inverse stereographic projection all points of $\cext$
can be projected onto the Riemann sphere.  As seen in \Fig{sterproj},
the complex plane is projected to the surface of the sphere along rays
originating from the north pole, and by convention the infinitely
remote point $\infty \in \cext$ is projected onto the north pole.  By
the same projection the roots $\{ z_1 , \ldots , z_n \}$ of the
Majorana polynomial \eqref{majpoly} are associated with the MPs on the
surface of the Majorana sphere. Therefore the Riemann sphere can be
employed as the Majorana sphere, and symmetric SLOCC operations have
the effect of transforming one set of Majorana roots (or equivalently
MPs) to another:

\begin{theorem}
  Two symmetric $n$ qubit states are SLOCC-equivalent iff there exists
  a M\"{o}bius transformation \eqref{mobiusform} between their
  Majorana roots.
\end{theorem}

M\"{o}bius transformations can be categorized into different types,
namely, \textit{parabolic, elliptic, hyperbolic} and
\textit{loxodromic} \cite{Knopp}, but a unifying feature is that two
(not necessarily diametral) points on the sphere are left invariant.
This generalizes the $\suc$ rotations where the diametrically opposite
intersections of the rotation axis with the sphere are left invariant.
The SLOCC operation from \Fig{ghztrafos}, mediated by the M\"{o}bius
transformation $f(z) = z/2$, is shown in detail in \Fig{sterproj}.
This transformation is hyperbolic, which means that the two invariant
points (here the north and south pole) act as attractive and repulsive
poles, with the MPs moving away from the repulsive pole towards the
attractive one.

A well-known property of M\"{o}bius transformations is that for any
two ordered sets of three pairwise distinct points $\{ v_1 , v_2 , v_3
\}$ and $\{ w_1 , w_2 , w_3 \}$ there always exists exactly one
M\"{o}bius transformation that maps one set to the other \cite{Knopp}.
With this it immediately becomes clear why DC classes $\mathd_{n_1 ,
  \ldots , n_d}$ with a diversity degree $d \leq 3$ consist of a
single SLOCC class \cite{Bastin09}.

In the following the three d.f. of M\"{o}bius transformations
\eqref{mobiusform} which genuinely belong to SLOCC operations (i.e.,
which cannot be realized by LOCC operations) are isolated, and a
visual interpretation in terms of the Majorana representation is
given.

\begin{theorem}\label{decomp}
  Every SLOCC operation between two symmetric $n$ qubit states can be
  factorized into an affine M\"{o}bius transformation of the form
  \begin{equation}\label{mobiusform2}
    \widetilde{f}(z) = A z + B \enspace , \enspace
    \text{with $A > 0$ , $B \in \mbbc$} \enspace ,
  \end{equation}
  and a LOCC operation. This decomposition is unique, and the set of
  transformations \eqref{mobiusform2} forms a group that is isomorphic
  to $\slc / \suc$.
\end{theorem}

\begin{proof}
  First the existence of a factorization of each SLOCC operation into
  a transformation $\widetilde{f}$ and a LOCC operation is shown.  For
  each $\mathb \in \slc$ we define $\widetilde{\mathb} =
  \lambda \mathb$ with $\lambda = \sqrt{a \cc{a} + c \cc{c}} >
  0$. Since $\widetilde{\mathb}$ describes the same SLOCC operation as
  $\mathb$, it suffices to show that $\widetilde{\mathb}$ can be
  decomposed into a LOCC operation $\matha \in \suc$ and a M\"{o}bius
  transformation of the form \eqref{mobiusform2}:
  \begin{equation*}
    \begin{pmatrix}
      \lambda a & \lambda b \\
      \lambda c & \lambda d
    \end{pmatrix} =
    \begin{pmatrix}
      \alpha & - \cc{\beta} \\
      \beta & \cc{\alpha}
    \end{pmatrix}
    \otimes
    \begin{pmatrix}
      A & B \\
      0 & 1
    \end{pmatrix} \enspace ,
  \end{equation*}
  with $A > 0$ and $\alpha , \beta , B \in \mbbc$, $\alpha \cc{\alpha}
  + \beta \cc{\beta} = 1$.  For given parameters $a,b,c,d \in \mbbc$
  with $ad - bc = 1$, this is fulfilled for $\alpha =
  \frac{a}{\lambda}$, $\beta = \frac{c}{\lambda}$, $A = \lambda^{2}$
  and $B = \frac{\lambda^{2} b + \cc{c}}{a} = \frac{\lambda^{2} d -
    \cc{a}}{c}$ (for $a=0$ or $c=0$ only one of the two identities
  holds). This proves the existence of a factorization.

  To show the uniqueness of factorizations, it is assumed that a given
  SLOCC operation $\mathb \in \slc$ can be factorized, up to scalar
  prefactors $\lambda_{1} , \lambda_{2} \in \mbbc \backslash \{ 0 \}$,
  in the above way by two sets of parameters $\{ \alpha_{1} ,
  \beta_{1} , A_{1} , B_{1} \}$ and $\{ \alpha_{2} , \beta_{2} , A_{2}
  , B_{2} \}$. Elimination of $\mathb$ from the resulting matrix
  equations yields the condition
  \begin{equation*}
    \frac{\lambda_2}{\lambda_1}
    \begin{pmatrix}
      \alpha_1 & - \cc{\beta_1} \\
      \beta_1 & \cc{\alpha_1}
    \end{pmatrix}
    \otimes
    \begin{pmatrix}
      A_1 & B_1 \\
      0 & 1
    \end{pmatrix} =
    \begin{pmatrix}
      \alpha_2 & - \cc{\beta_2} \\
      \beta_2 & \cc{\alpha_2}
    \end{pmatrix}
    \otimes
    \begin{pmatrix}
      A_2 & B_2 \\
      0 & 1
    \end{pmatrix}
     .
  \end{equation*}
  A straightforward calculation yields $\lvert
  \frac{\lambda_{2}}{\lambda_{1}} \rvert = 1$, and from this it
  readily follows that the two sets of parameters must coincide.  This
  uniqueness implies that the set of transformations $\widetilde{f}$
  is isomorphic to $\slc / \suc$, and their group properties are
  easily verified explicitly.
\end{proof}

\Theo{decomp} is closely related to the polar decomposition of
matrices which states that every invertible complex matrix can be
uniquely decomposed into a unitary matrix and a positive-semidefinite
Hermitian matrix \cite{Golub}.  However, while the matrices of the
affine transformations $\widetilde{f}$ are positive, they are in
general not Hermitian, and the introduction of the prefactor $\lambda$
in the proof is necessary because $\matha$ and $\mathb$ are defined to
have unit determinants.

\begin{figure}[ht]
  \includegraphics{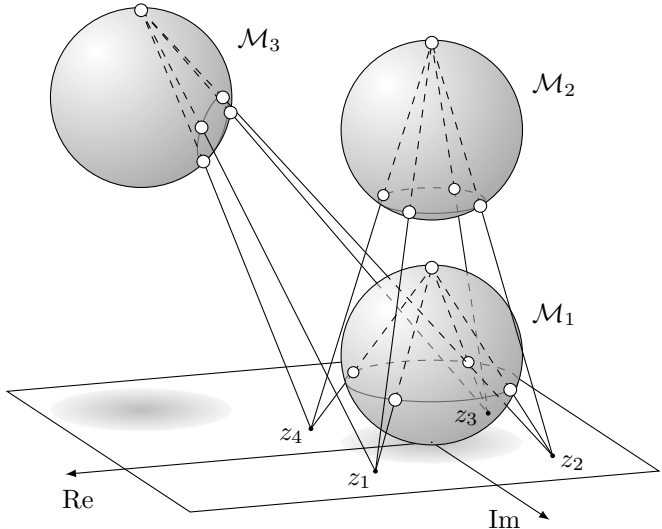}
  \caption{\label{translations} Alternative visualization of
    M\"{o}bius transformations where a fixed set of complex points is
    projected onto the surface of a moving sphere.  The three innate
    freedoms of SLOCC operations not present in LOCC operations are
    then described by the translations of the Majorana sphere in
    $\mbbrr$.  The north pole of sphere $\mathm_{1}$ (with the MP
    distribution of the 5 qubit ``square pyramid'' state outlined in
    \protect\cite{Aulbach10}) lies 2 units above the origin of the
    complex plane, while the one of $\mathm_{2}$ lies 5 units above,
    and $\mathm_{3}$ is additionally displaced horizontally by a
    vector $5 - 5 \I$.  The parameters $( A , B )$ of \Eq{mobiusform2}
    for the transformation of $\mathm_{1}$ to $\mathm_{2}$ and
    $\mathm_{3}$ are $( \frac{5}{2} , 0 )$ and $( \frac{5}{2} , 5 - 5
    \I )$, respectively.}
\end{figure}

The orthodox way to visualize M\"{o}bius transformations is to fix the
Riemann sphere in $\mbbrr$ (usually with the sphere's center or south
pole coinciding with the complex plane's origin), and points $\{ z_1 ,
\ldots , z_n \}$ on the complex plane are transformed to different
points $\{ z'_1 , \ldots , z'_n \}$ under the action of the functions
\eqref{mobiusform}.  By means of the inverse stereographic projection,
this transformation can then be observed on the sphere too, as seen in
\Fig{sterproj}.

Alternatively, the points in the plane can be considered fixed, and
instead the Riemann sphere moves in $\mbbrr$, as shown in
\Fig{translations}.  The six d.f. of the M\"{o}bius transformations
are then split into three translational freedoms (movement of sphere
in $\mbbrr$) and three rotational freedoms (rotation of sphere around
its axes).  By considering these elementary operations it can be
verified by calculation that this is an equivalent way of viewing the
change of points on the sphere under the action of M\"{o}bius
transformations.  In this approach the affine transformations
\eqref{mobiusform2} are easily identified as the set of all
translations in $\mbbrr$ which leave the sphere's north pole above the
complex plane. A general SLOCC operation between symmetric states can
therefore be described as a translation of the Majorana sphere in
$\mbbrr$, followed by a rotation.  The parameters of the affine
function $\widetilde{f} (z) = A z + B$ are connected to the
translation as follows: The parameter $A = \frac{h_{2}}{h_{1}}$ is the
ratio of the heights of the north pole before ($h_{1}$) and after
($h_{2}$) the transformation, and $B$ is the horizontal displacement
vector (cf. \Fig{translations}).

\section{Representative States for Symmetric Entanglement
  Classes}\label{representative}

Multiqubit entanglement classes have been well studied before, in
particular, the SLOCC-equivalent classes for a single copy of a pure
$n$ qubit state.  For 2 qubits every entangled state can be turned
into a singlet by a SLOCC operation, while for 3 qubits there exist
three classes with non-symmetric bipartite entanglement as well as two
classes for GHZ-type and W-type entanglement \cite{Dur00,Acin00}.  For
as few as 4 qubits, however, the number of SLOCC classes becomes
infinite \cite{Dur00}.  Verstreate et \emph{al.}  \cite{Verstraete02}
suggested to solve this dilemma by identifying nine different
\emph{families} of 4 qubit entanglement, and a similar approach was
pursued by Lamata et \emph{al.}  \cite{Lamata07}.  In the symmetric
sector a different approach is the introduction of the aforementioned
DC classes.

\begin{figure}[ht]
  \includegraphics{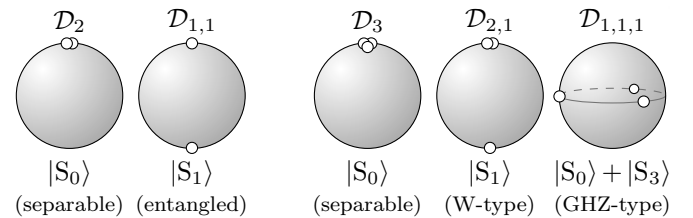}
  \caption{\label{slocc23} All DC classes of 2 and 3 qubit symmetric
    states are listed together with representative states and their MP
    distribution. Each DC class is also a SLOCC class, which implies
    that every state of a DC class can be reached from the
    representative state by a SLOCC operation.}
\end{figure}

In the following the SLOCC and DC classes of symmetric states of up to
$5$ qubits are characterized, and representative states with simple MP
distributions are given for each equivalence class.  Since all DC
classes of 2 and 3 qubit states have a diversity degree of 3 or less,
their DC classes are identical to SLOCC classes.  In \Fig{slocc23}
these classes are listed together with a representative state for each
class.  For three qubits the class $\mathd_{3}$ contains the separable
states, $\mathd_{2,1}$ the W-type entangled states and
$\mathd_{1,1,1}$ the GHZ-type entangled states.

\begin{figure}
  \includegraphics{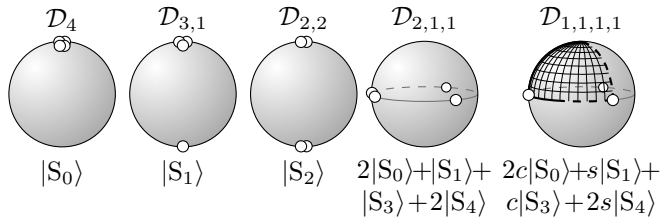}
  \caption{\label{slocc4} Only four of the five DC classes of 4 qubit
    symmetric states coincide with a single SLOCC class. Due to the
    continuum of SLOCC classes in $\mathd_{1,1,1,1}$, only three MPs
    can be fixed in its representative state, with the unique
    locations for the fourth MP $c \ket{0} + s \ket{1}$ parameterizing
    the set of representative states. Here $c = \cos \frac{\theta}{2}$
    and $s = \E^{\I \varphi} \sin \frac{\theta}{2}$, and the range of
    parameters is $( \theta , \varphi ) \in \bmr{ \{ } [0,
    \frac{\pi}{2}) \times [0, \frac{2 \pi}{3}) \bmr{ \} } \cup \bmr{
      \{ } \{ \frac{\pi}{2} \} \times (0 , \frac{\pi}{3}] \bmr{ \} }$,
    shown as a black grid.  The fixed equatorial MPs of the
    representative states are equidistantly spaced.}
\end{figure}

For symmetric states of 4 qubits there exist five DC classes and a
continuum of SLOCC classes \cite{Bastin09}.  As shown in \Fig{slocc4},
four of the DC classes coincide with SLOCC classes, while the generic
class $\mathd_{1,1,1,1}$ is comprised of a continuum of SLOCC classes
(cf. Figure 2 in \cite{Markham10}).  The high symmetry present in an
equidistant distribution of three MPs along the equator facilitates
the restriction of the remaining parameters of the generic SLOCC
classes to a well-defined, connected area on the sphere's surface:

\begin{theorem}\label{theorem4qubit}
  Every symmetric state of 4 qubits is SLOCC-equivalent to exactly one
  state of the set \vspace{2mm} \\
  $\{ \sym{0}, \enspace \sym{1}, \enspace \sym{2}, \enspace 2 \sym{0}
  + t \sym{1} + \sym{3} + 2t
  \sym{4} \} \, , \text{ with}$ \vspace{2mm} \\
  $t = \E^{\I \varphi} \tan \frac{\theta}{2}$, and $( \theta , \varphi
  ) \! \in \! \bm{ \{ } [0, \frac{\pi}{2}) \! \times \! [0, \frac{2
    \pi}{3}) \bm{ \} } \cup \bm{ \{ } \{ \frac{\pi}{2} \} \! \times \!
  [0 , \frac{\pi}{3}] \bm{ \} }.$
\end{theorem}

\begin{proof}
  First it will be shown that every symmetric 4 qubit state $\psis$
  can be transformed by SLOCC into one of the above states.  From the
  previous discussion and \Fig{slocc4}, this is clear for all DC
  classes except $\mathd_{1,1,1,1}$.  Given an arbitrary state $\psis
  \in \mathd_{1,1,1,1}$, there always exists a M\"{o}bius
  transformation $f: \psis \rightarrow \ket{\psi '^{\text{s}}}$ s.t.
  three of the distinct MPs are projected onto the three corners of an
  equilateral triangle in the equatorial plane. If the fourth MP
  $\ket{\phi_4}$ is not projected into the area parameterized by $(
  \theta , \varphi ) \in \bm{ \{ } [0, \frac{\pi}{2}) \times [0,
  \frac{2 \pi}{3}) \bm{ \} } \cup \bm{ \{ } \{ \frac{\pi}{2} \} \times
  (0 , \frac{\pi}{3}] \bm{ \} }$ (cf. \Fig{slocc4}), then it can be
  projected into that area through a combination of $\{
  \text{R}_{\text{x}} ( \pi ) , \text{R}_{\text{z}} ( \frac{2 \pi}{3}
  ) \}$-rotations of the Majorana sphere (which preserve the
  equatorial MP distribution).

  It remains to show that this set of states is unique, i.e., two
  different MPs $\ket{\phi_4}$ and $\ket{\phi_4 '}$ within the
  aforementioned parameter range give rise to two different states
  $\psis \neq \ket{\psi '^{\text{s}}}$ which are
  SLOCC-inequivalent. This can be verified by the cross-ratio
  preservation of M\"{o}bius transformations \cite{Knopp}, namely,
  that a projection of an ordered quadruple of distinct complex
  numbers $\{ v_1, v_2, v_3, v_4 \}$ onto another quadruple $\{ w_1,
  w_2, w_3, w_4 \}$ requires that
  \begin{equation*}
    \frac{(v_1 - v_3)(v_2 - v_4)}{(v_2 - v_3)(v_1 - v_4)} = 
    \frac{(w_1 - w_3)(w_2 - w_4)}{(w_2 - w_3)(w_1 - w_4)} \enspace .
  \end{equation*}
  By considering all 4! possible projections between the MPs of
  $\psis$ and $\ket{\psi '^{\text{s}}}$ it can be explicitly verified
  that a transformation is possible only if $\ket{\phi_4} =
  \ket{\phi_4 '}$.
\end{proof}

\begin{figure}
  \includegraphics{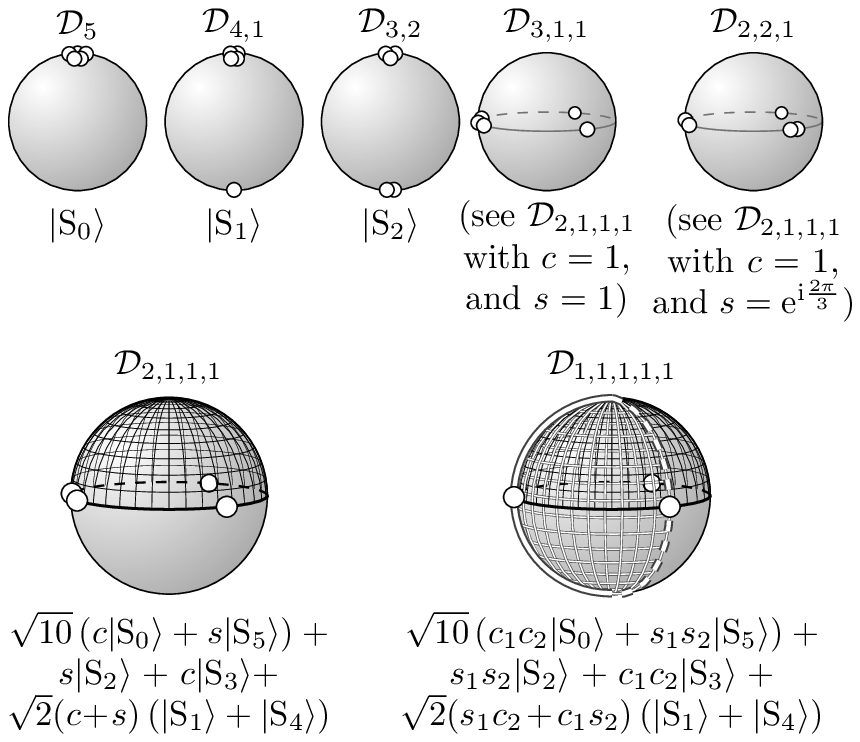}
  \caption{\label{slocc5} The first five of the seven DC classes of 5
    qubit symmetric states coincide with SLOCC classes, while the
    representative states of $\mathd_{2,1,1,1}$ are parameterized by
    one MP $c \ket{0} + s \ket{1}$ (black grid), and those of
    $\mathd_{1,1,1,1,1}$ by two MPs (black and white grid).  The
    parameter range for $\mathd_{2,1,1,1}$ is $( \theta , \varphi )
    \in \bmr{ \{ } [0, \frac{\pi}{2} ) \times [0, 2 \pi) \bmr{ \} }
    \cup \bmr{ \{ } \{ \frac{\pi}{2} \} \times (0 , \pi ] \bmr{ \} }
    \backslash \bmr{ \{ } \{ \frac{\pi}{2} \} \times \{ \frac{2
      \pi}{3} \} \bmr{ \} }$.  For $\mathd_{1,1,1,1,1}$ the range of
    $( \theta_1 , \varphi_1 )$ is the same as $( \theta , \varphi )$,
    and $( \theta_2 , \varphi_2 ) \in \bmr{ \{ } [0, \pi ] \times [0,
    \frac{2 \pi}{3}) \bmr{ \} } \backslash \bmr{ \{ } \{ \frac{\pi}{2}
    \} \times \{ 0 \} \bmr{ \} }$.  The fixed equatorial MPs of the
    representative states are all equidistantly spaced.}
\end{figure}

The DC classes of 5 qubits and representative states for the SLOCC
classes can be seen in \Fig{slocc5}. The SLOCC classes of the generic
class $\mathd_{1,1,1,1,1}$ can be parameterized by two complex
variables, corresponding to two MPs in the black and white area,
respectively.  Unlike the 4 qubit case, however, this parameterization
is neither unique, nor confined to the generic DC class.  Different
sets of parameters $( \theta_1 , \varphi_1 , \theta_2 , \varphi_2 )
\neq ( \theta'_1 , \varphi'_1 , \theta'_2 , \varphi'_2 )$ can give
rise to SLOCC-equivalent states, and for $( \theta_1 , \varphi_1 ) = (
\theta_2 , \varphi_2 )$ the corresponding state does not even belong
to $\mathd_{1,1,1,1,1}$ because of coinciding MPs.  A unique set of
representative states is therefore provided only for the subset of
symmetric states with a MP degeneracy:

\begin{theorem}\label{theorem5qubitdeg}
  Every symmetric state of 5 qubits with a MP degeneracy (i.e.,
  diversity degree $d < 5$) is SLOCC-equivalent to exactly one state
  of the set
  \vspace{2mm} \\
  $\{ \sym{0}, \: \: \sym{1}, \: \: \sym{2}, \: \: \sqrt{10} \left(
    \sym{0} + t \sym{5} \right) + t \sym{2} + \sym{3} + \sqrt{2}
  \left( 1 + t \right) \left( \sym{1} + \sym{4} \right) \} \, , \text{
    with } t = \E^{\I \varphi} \tan \frac{\theta}{2}$, and
  \vspace{2mm} \\
  $( \theta , \varphi ) \in \bm{ \{ } [0, \frac{\pi}{2} ) \times [0, 2
  \pi) \bm{ \} } \cup \bm{ \{ } \{ \frac{\pi}{2} \} \times [0 , \pi ]
  \bm{ \} }$ .
\end{theorem}

\begin{proof}
  The proof runs analogous to the one of \Theo{theorem4qubit}, with
  the observation that the representative states of the
  $\mathd_{3,1,1}$ and $\mathd_{2,2,1}$ class are readily subsumed in
  the parameter range of $\mathd_{2,1,1,1}$.  The fixed MPs of
  $\mathd_{2,1,1,1}$ are left invariant under a $\text{R}_{\text{x}} (
  \pi )$-rotation, thus ensuring that the remaining MP can be
  projected into the desired parameter range. The uniqueness is again
  verified by considering all possible cross-ratios.
\end{proof}

An over-complete set of representative states for the general case can
then be given as follows:

\begin{corollary}\label{theorem5qubit}
  Every symmetric state of 5 qubits is SLOCC-equivalent to one or more
  state of the set
  \vspace{2mm} \\
  $\{ \sym{0}, \: \: \sym{1}, \: \: \sym{2}, \: \: \sqrt{10} \left(
    \sym{0} + t_1 t_2 \sym{5} \right) + t_1 t_2 \sym{2} + \sym{3} +
  \sqrt{2} \left( t_1 + t_2 \right) \left( \sym{1} + \sym{4} \right)
  \} \, , \text{ with } t_i = \E^{\I \varphi_i} \tan
  \frac{\theta_i}{2} \, , \text{ and}$
  \vspace{2mm} \\
  $( \theta_1 , \varphi_1 ) \in \bm{ \{ } [0, \frac{\pi}{2}] \times
  [0, 2 \pi) \bm{ \} } \cup \bm{ \{ } \{ \frac{\pi}{2} \} \times (0 ,
  \pi ] \bm{ \} }$ , \vspace{1mm} \\ $( \theta_2 , \varphi_2 ) \in
  \bm{ \{ } [0, \pi] \times [0, \frac{2 \pi}{3}) \bm{ \} }$ .
\end{corollary}

\begin{proof}
  Only the generic class $\mathd_{1,1,1,1,1}$ needs to be considered.
  Given an arbitrary state of this class, three of its MPs can be
  projected onto the corners of an equilateral triangle by means of a
  M\"{o}bius transformation. These MPs are left invariant under $\{
  \text{R}_{\text{x}} ( \pi ) , \text{R}_{\text{z}} ( \frac{2 \pi}{3}
  ) \}$-rotations. If the fourth MP does not lie in the $( \theta_1 ,
  \varphi_1 )$-area, it can be projected there by a
  $\text{R}_{\text{x}} ( \pi )$-rotation. Subsequent
  $\text{R}_{\text{z}} ( \frac{2 \pi}{3} )$-rotations can project the
  fifth MP into the $( \theta_2 , \varphi_2 )$-area, while leaving the
  fourth MP in the $( \theta_1 , \varphi_1 )$-area.
\end{proof}

As the number of qubits increases, the picture gradually becomes more
complicated, because DC classes with diversity degree $n$ contain a
continuous range of SLOCC classes that is parameterized by $n-3$
variables \cite{Bastin09}.

\section{Applications and Connections}\label{applications}

\subsection{Four qubit entanglement families}
\label{families}

To describe the behavior of 4 qubit states under SLOCC operations, the
concept of entanglement families (EF) was introduced in
\cite{Verstraete02}.  Nine different EFs were identified, and every 4
qubit state is SLOCC-equivalent to one of these families.  Hence,
SLOCC is a refinement of the entanglement families: SLOCC $<$ EF.

It will now be determined in which EFs the symmetric SLOCC and DC
classes are located.  The separable state $\sym{0}$, and therefore the
entire $\mathd_{4}$ class, is present (up to LUs) in the family
$L_{{abc}_{2}}$, namely, by setting $a=b=c=0$.  The W state $\sym{1}$
is LU-equivalent to the family $L_{{ab}_{3}}$ for $a=b=0$.  The state
$\sym{2}$ can be found in the general family $G_{abcd}$ by setting
$a=1, b=2, c=0, d=-1$.  The continuum of SLOCC classes present in the
generic family $\mathd_{1,1,1,1}$ has previously been parameterized in
\cite{Bastin09} as $( \sym{0} + \sym{3} ) + \mu \sym{2}$, with $\mu
\in \mbbc \backslash \{ \pm \frac{1}{\sqrt{3}} \}$.  These states are
easily recovered from the general family $G_{abcd}$ with $a = 1 +
\frac{\mu}{2}, b = \mu, c = 0, d=1 - \frac{\mu}{2}$.

It is noteworthy that two different DC classes, namely, $\mathd_{2,2}$
and $\mathd_{1,1,1,1}$, belong to the same entanglement family
$G_{abcd}$.  On the other hand, all states of a given DC class belong
to only one EF \footnote{It is not known yet to which EF the DC class
  $\mathd_{2,1,1}$ belongs. However, since $\mathd_{2,1,1}$ consists
  of only one SLOCC class, all of its states must belong to a single
  EF.}.  Thus \Theo{hierarchy} can be stated more precisely for the
four qubit case:

\begin{theorem}\label{hierarchy4qubit}
  The symmetric subspace of the pure 4 qubit Hilbert space has the
  following refinement hierarchy of entanglement partitions:
  \begin{equation}
    \text{LOCC} < \text{SLOCC} < \text{DC} < EF \enspace .
  \end{equation}
\end{theorem}

\subsection{Determination of SLOCC inequivalence from the MP
  distribution}
\label{determiningslocc}

The known properties of M\"{o}bius transformations can be utilized to
determine from the MP distributions whether symmetric states with the
same degeneracy configuration could be SLOCC-equivalent. For example,
circles on the surface of the Majorana sphere are always projected
onto circles \cite{Knopp}, and this trait can be exploited by looking
for circles with a certain number of MPs.

\begin{figure}[ht]
  \includegraphics{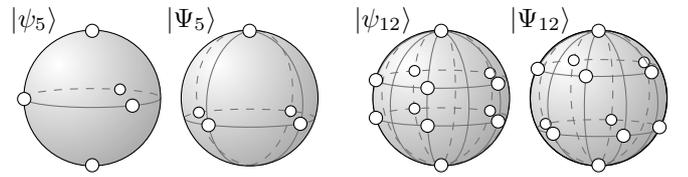}
  \caption{\label{maxentstates} Four highly or maximally entangled
    symmetric states which were introduced in
    \protect\cite{Aulbach10}.  The 5 qubit ``trigonal bipyramid''
    state $\ket{\psi_{5}}$ is SLOCC-inequivalent to the maximally
    entangled symmetric state of 5 qubits $\ket{\Psi_{5}}$. Likewise,
    the 12 qubit ``icosahedron'' state $\ket{\Psi_{12}}$ cannot be
    reached from $\ket{\psi_{12}}$ by SLOCC operations.}
\end{figure}

As an example, the two 5 qubit states shown in \Fig{maxentstates} are
not SLOCC-equivalent, because $\ket{\Psi_5}$ exhibits a ring with 4
MPs, while such a ring is not present in $\ket{\psi_5}$.  Similarly
one can show that for the maximally entangled symmetric states (in
terms of the geometric measure) of 10 and 11 qubits, as discussed in
\cite{Aulbach10}, the presumed solutions for the general case are not
SLOCC-equivalent to those for the subset of states with positive
coefficients.  For 12 qubits it is not as obvious that the general and
positive solutions, shown in \Fig{maxentstates}, are
SLOCC-inequivalent, since both states have several rings with 4 or 5
MPs each.  For the highly symmetric icosahedron state
$\ket{\Psi_{12}}$ it is possible to identify twenty different circles,
each through three adjacent MPs (the corners of all faces of the
icosahedron), so that the interior of each circle contains no
MPs. This property must be preserved under M\"{o}bius transformations,
but for $\ket{\psi_{12}}$ it is not possible to find such twenty
distinct circles that are all free of other MPs in their interior.

\section{Conclusion}\label{conclusion}

In this paper the three entanglement classification schemes LOCC,
SLOCC and Degeneracy Configuration were employed to characterize and
explore symmetric multiqubit states. It was found that the M\"{o}bius
transformations from complex analysis do not only allow for a simple
and complete description of the freedoms present in SLOCC operations,
but also provide a straightforward visualization of these freedoms by
means of the Majorana sphere.  In particular, it would be promising to
study how the entanglement and interconversion probabilities changes
under the action of the M\"{o}bius transformations $\widetilde{f}$
which translate the Majorana sphere in $\mbbrr$.  For example, in
\Fig{translations} the maximally entangled symmetric 5 qubit state in
terms of the geometric measure is displayed at $\mathm_{1}$, and any
translation of the sphere decreases the entanglement of the underlying
state.

The symmetric SLOCC classes of up to 5 qubits were fully characterized
by representative states whose MP distributions are of a particularly
simple form, or can be easily parameterized by well-defined areas on
the sphere for the variable MPs.  For 4 qubits the concept of
entanglement families was fitted into the hierarchy of symmetric
entanglement classification schemes, and it was demonstrated how the
existing theory of M\"{o}bius transformations can prove helpful to
easily determine whether two symmetric states are SLOCC-equivalent or
not.

\begin{acknowledgments}
  The author thanks D~Markham, J~Biamonte, V~Vedral, J~Dunningham and
  M~Williamson for very helpful discussions.  This work is supported
  by the National Research Foundation \& Ministry of Education,
  Singapore.

  \emph{Note added.} During the completion of this paper I became
  aware of a similar work which also points out the relationship
  between symmetric SLOCC operations and M\"{o}bius transformations
  \cite{Ribeiro11}.
\end{acknowledgments}

\end{document}